\documentclass[11pt]{article}

\usepackage{amsmath, amssymb, amsfonts}
\usepackage{geometry}
\usepackage{hyperref}
\usepackage{booktabs}
\usepackage{graphicx}
\usepackage{subcaption}
\usepackage{wrapfig}
\usepackage{parskip} 
\usepackage{tikz}
\usetikzlibrary{arrows.meta,positioning,fit,calc,backgrounds}

\usepackage{xcolor}

\usepackage{hyperref}
\hypersetup{
    colorlinks=true,      
    citecolor=orange,     
    linkcolor=blue,       
    urlcolor=blue         
}

\usepackage{amsthm}
\newtheorem{lemma}{Lemma}
\newtheorem{proposition}{Proposition}
\newtheorem{theorem}{Theorem}

\title{Zero-shot adaptation to order book dynamics\thanks{Preprint. Draft version.}}
\author{Arip Asadulaev \\ \texttt{arip.asadulaev@mbzuai.ac.ae}}
\date{}

\begin{document}

\maketitle

\begin{abstract}
We describe an adaptive market-making architecture that preserves the analytical structure of the Avellaneda--Stoikov framework while introducing a successor measure-style adaptation mechanism. In our paper we keep Avellaneda--Stoikov fast Hamilton--Jacobi--Bellman structure and make it adaptive to changing market regimes and trading objectives. The central idea is to separate market dynamics from the trading objective. The market state determines a low-dimensional set of Avellaneda--Stoikov parameters, while recent realized rewards determine a low-dimensional objective vector. The HJB forward map then converts this objective into optimal bid and ask quotes through a scalarization of future reward features.
\end{abstract}
\section{Introduction}
\begin{wrapfigure}{r}{0.55\textwidth}
  \vspace{-5mm}
  \includegraphics[width=0.55\textwidth]{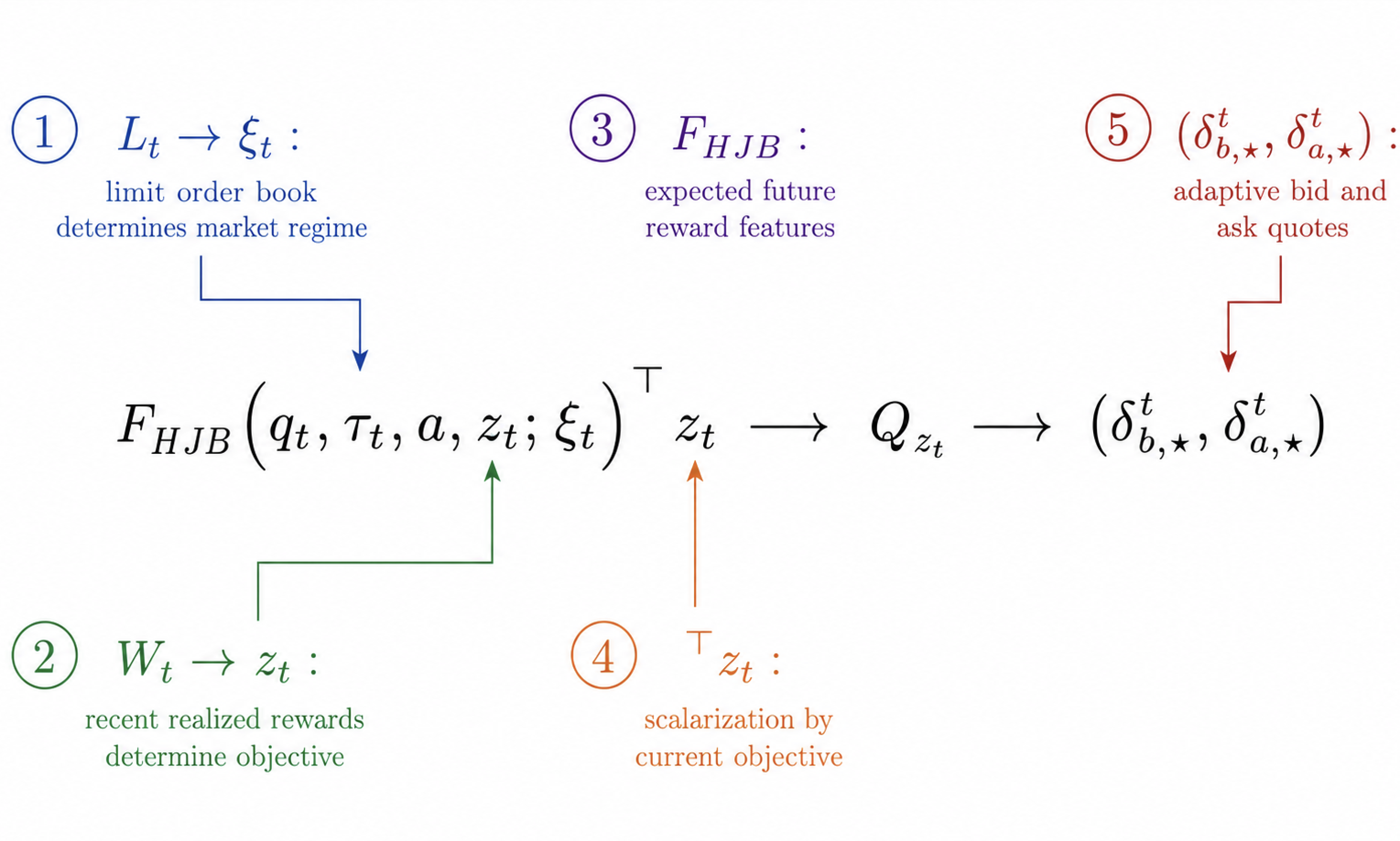}
  \caption{
    Conceptual of the proposed adaptive market making. The market state determines local dynamics, recent realized rewards determine the adaptive objective, and the HJB-forward representation converts both into bid and ask quotes.
    }
  \vspace{-5mm}
\end{wrapfigure}
Classical Avellaneda--Stoikov market making \cite{avellaneda2008high} is built around a small number of interpretable parameters: volatility, order arrival intensities, inventory risk aversion, and time to horizon. This makes the model fast and structurally robust. However, the objective is usually fixed in advance. If the desired inventory target, risk appetite, or adverse-selection penalty changes, the model must be reinterpreted or recalibrated.

Recently, in AI and especially in the reinforcement learning community, a range of zero-shot algorithms such as Forward-Backward representations \cite{touati2021learning, tirinzoni2025zero, asadulaev2026zero} has been developed. The core concept of these methods is a \emph{separation of the general environment dynamics from the goals you want to achieve based on these given dynamics}. In our paper, we show that this point of view transfers well to the market problem \cite{touati2021learning}.

We propose an approach that keeps the Avellaneda--Stoikov dynamics as the physical core, but replaces the fixed objective by an adaptive reward vector. The model has two low-dimensional adaptive components. The first one estimates the current market regime, denoted by $\xi_t$, from the limit order book state. The second one estimates the current trading objective, denoted by $z_t$, from recent realized rewards. The policy is then obtained by solving or approximating an Hamilton-Jacobi-Bellman (HJB)-style forward representation\cite{touati2021learning} and scalarizing it by the current objective vector (See Figure 1):
\[
Q_{z_t}(s_t,a)
=
F_{\mathrm{HJB}}(q_t,\tau_t,a,z_t;\xi_t)^\top z_t.
\]
Our $F_{\mathrm{HJB}}$ is fully based on the Avellaneda--Stoikov structure, we will see it later. The scalar action-value is only the inner product between the HJB successor-feature representation and the current objective vector. The complete architecture can be summarized as
\[
L_t
\longmapsto
\xi_t,
\qquad
\mathcal W_t
\longmapsto
z_t,
\qquad
(\xi_t,z_t)
\longmapsto
F_{\mathrm{HJB}}^\top z_t
\longmapsto
(\delta_t^{b,\star},\delta_t^{a,\star}).
\]
Here \(L_t\) is the current market microstructure state and \(\mathcal W_t\) is a recent window of completed observations with realized reward labels.

\section{Market State and Avellaneda--Stoikov Dynamics}

The state of the market maker is $s_t=(m_t,q_t,X_t,\tau_t,L_t)$, where \(m_t\) is the mid-price, \(q_t\) is inventory, \(X_t\) is cash, \(\tau_t=T-t\) is time to horizon, and \(L_t\) contains limit order book features such as spread, imbalance, depth, recent trade flow, realized volatility, queue information, and toxicity indicators. The action is $a_t=(\delta_t^b,\delta_t^a,v_t^b,v_t^a)$, where the bid and ask quotes are $p_t^b=m_t-\delta_t^b$, and $p_t^a=m_t+\delta_t^a$. In the simplest and most Avellaneda--Stoikov-like version, the order sizes are fixed, $v_t^b=v_t^a=v$, and the control is reduced to the choice of bid and ask distances
\[
a_t=(\delta_t^b,\delta_t^a).
\]
The mid-price follows the Avellaneda--Stoikov diffusion $dm_t=\sigma_t dW_t$. Bid and ask executions are modeled through exponential arrival intensities
\[
\lambda_t^b(\delta^b)
=
A_t^b e^{-\kappa_t^b\delta^b},
\qquad
\lambda_t^a(\delta^a)
=
A_t^a e^{-\kappa_t^a\delta^a}.
\]
The difference from the classical model is that the parameters are no longer fixed constants. They are inferred from the current market state: $
\xi_t=\Xi_\psi(L_t)$, with
\[
\xi_t
=
\left(
\sigma_t,
A_t^b,
A_t^a,
\kappa_t^b,
\kappa_t^a,
c_t^b,
c_t^a
\right).
\]
The quantities \(c_t^b\) and \(c_t^a\) are local adverse-selection costs. If buying through a bid fill is currently toxic, then \(c_t^b\) increases. If selling through an ask fill is currently toxic, then \(c_t^a\) increases. These costs enter the quote equations directly and push dangerous quotes farther away from the mid-price. The inventory transition is determined by the executed bid and ask orders:
\[
q_{t+1}
=
q_t
+
v_t^b N_{t+1}^b
-
v_t^a N_{t+1}^a,
\]
where \(N_{t+1}^b\) and \(N_{t+1}^a\) are execution indicators. The mark-to-market wealth is $W_t=X_t+q_t m_t$,
and the one-step PnL is $\Delta W_{t+1}=W_{t+1}-W_t$.

\section{Linear Return Representation and Adaptive Objective}

The reward is represented as a linear functional of outcome features:
\[
r_z(s,a,s')
=
z^\top \phi(s,a,s').
\]
A minimal feature vector for adaptive market making is
\[
\phi(s_t,a_t,s_{t+1})
=
\begin{pmatrix}
\Delta W_{t+1}\\
q_{t+1}\\
q_{t+1}^2\\
\mathrm{adv}_{t+1}
\end{pmatrix},
\]
where \(\mathrm{adv}_{t+1}\) is an adverse-selection or markout-loss. The corresponding objective vector is
\[
z_t
=
\begin{pmatrix}
z_{\mathrm{pnl},t}\\
z_{q,t}\\
z_{q^2,t}\\
z_{\mathrm{adv},t}
\end{pmatrix}.
\]

A target-inventory objective is recovered as a special case. Suppose that the desired inventory is \(\theta_t\), the inventory penalty is \(\lambda_t\), and the adverse-selection penalty is \(\nu_t\). The reward is
\[
r_{\theta_t}
=
\Delta W
-
\lambda_t(q'-\theta_t)^2
-
\nu_t \mathrm{adv}.
\]
Expanding the quadratic term gives
$
-\lambda_t(q'-\theta_t)^2
=
-\lambda_t (q')^2
+
2\lambda_t\theta_t q'
-
\lambda_t\theta_t^2
$. The last term does not depend on the action and can be ignored for control. Hence the reward is equivalent, for action selection, to $
r_{\theta_t}
\sim
\Delta W
+
2\lambda_t\theta_t q'
-
\lambda_t(q')^2
-
\nu_t\mathrm{adv}
$. Therefore the target-inventory objective corresponds to
\[
z(\theta_t,\lambda_t,\nu_t)
=
\begin{pmatrix}
1\\
2\lambda_t\theta_t\\
-\lambda_t\\
-\nu_t
\end{pmatrix}.
\]

Conversely, if an adaptive vector \(z_t\) has already been estimated from data, one can recover the implied risk penalty and target inventory from its inventory components. Since
\[
z_{q,t}=2\lambda_t\theta_t,
\qquad
z_{q^2,t}=-\lambda_t,
\]
we obtain
\[
\lambda_t=-z_{q^2,t},
\qquad
\theta_t
=
\frac{z_{q,t}}{-2z_{q^2,t}}.
\]
Thus the target inventory need not be imposed externally. It can be inferred from recent realized rewards. Our adaptation is introduced by estimating \(z_t\) from recent data. Let
\[
\mathcal W_t
=
\{i:t-H_{\mathrm{roll}}\le i\le t-H\}
\]
be a recent window of observations for which the reward can already be labeled on a markout horizon \(H\). For each past observation, define a realized reward label
\[
r_i^H
=
N_i^b v_i^b
\left(
\delta_i^b+m_{i+H}-m_i
\right)
+
N_i^a v_i^a
\left(
\delta_i^a-(m_{i+H}-m_i)
\right)
-
\mathrm{fees}_i
-
\mathrm{impact}_i.
\]
The backward representation
\[
B_\psi(s_i,a_i)\in\mathbb R^d
\]
maps a past state-action pair into the same reward-coordinate space as \(z_t\). The adaptive objective is then estimated by
\[
\widehat z_t
=
\sum_{i\in\mathcal W_t}
w_{t,i}
r_i^H B_\psi(s_i,a_i),
\]
where the weights may be exponentially decaying:
\[
w_{t,i}
=
\frac{
\exp(-(t-i)/\ell)
}{
\sum_{j\in\mathcal W_t}\exp(-(t-j)/\ell)
}.
\]
Equivalently,
\[
\widehat z_t
=
\mathbb E_{\rho_t}
\left[
r^H(s,a)B_\psi(s,a)
\right],
\]
where \(\rho_t\) is the recent local empirical distribution. This is the FB adaptation step. Recent realized rewards select the reward direction that is currently being paid by the market.

\section{HJB Forward Representation and Scalarization}

The HJB-forward representation is the central value object of the model: $F_{\mathrm{HJB}}(q,\tau,a,z;\xi)\in\mathbb R^d$. It represents expected future reward features under the Avellaneda--Stoikov dynamics with market parameters \(\xi\), when the current action is \(a\) and future actions are optimal for the objective \(z\):
\[
F_{\mathrm{HJB}}(q,\tau,a,z;\xi)
=
\mathbb E^{\pi_z,\xi}
\left[
\sum_{k\ge 0}
\beta^k
\phi_{t+k}
\ \middle|\
q_t=q,\ a_t=a
\right].
\]
The scalar action-value is obtained by the inner product
\[
Q_z(q,\tau,a;\xi)
=
F_{\mathrm{HJB}}(q,\tau,a,z;\xi)^\top z.
\]
The policy is
\[
\pi_z(q,\tau;\xi)
=
\arg\max_{a\in\mathcal A_{\mathrm{safe}}}
F_{\mathrm{HJB}}(q,\tau,a,z;\xi)^\top z.
\]

This formulation keeps the nonlinearity of the HJB problem because the forward representation is indexed by \(z\). The future policy depends on the objective, and therefore the successor features depend on the objective as well. The model does not assume that a single objective-independent successor map is sufficient for all reward functions. Instead, it keeps the Avellaneda--Stoikov HJB structure and uses \(z\) as the low-dimensional control parameter that changes the objective. Let $U_z(\tau,q;\xi)\in\mathbb R^d$ be the vector value function. The vector Bellman recursion is
\[
F_{\mathrm{HJB}}(\tau,q,a,z;\xi)
=
\bar\phi(q,a;\xi)
+
\beta
\mathbb E_{\xi}
\left[
U_z(\tau-\Delta t,q')
\mid q,a
\right],
\]
and the state value satisfies
\[
U_z(\tau,q;\xi)
=
F_{\mathrm{HJB}}(\tau,q,\pi_z(\tau,q;\xi),z;\xi).
\]
The scalar HJB value is only a projection:
\[
h_z(\tau,q;\xi)
=
U_z(\tau,q;\xi)^\top z.
\]
Thus the model may be viewed as a vector-valued HJB in reward-feature space, followed by a scalarization using the current adaptive objective.

\section{Closed-Form Avellaneda--Stoikov Quoting Rule}

Under exponential fill intensities and fixed order size \(v\), the optimal quote distances can be written in an Avellaneda--Stoikov-like closed form. Define the scalar continuation value
\[
h_t(q)
=
h_{z_t}(\tau_t,q;\xi_t)
=
U_{z_t}(\tau_t,q;\xi_t)^\top z_t.
\]
A bid fill moves inventory from \(q_t\) to \(q_t+v\). An ask fill moves inventory from \(q_t\) to \(q_t-v\). The optimal bid distance is
\[
\delta_t^{b,\star}
=
\frac{1}{\kappa_t^b}
-
\frac{
h_t(q_t+v)-h_t(q_t)
}{v}
+
\frac{c_t^b}{v}.
\]
The optimal ask distance is
\[
\delta_t^{a,\star}
=
\frac{1}{\kappa_t^a}
-
\frac{
h_t(q_t-v)-h_t(q_t)
}{v}
+
\frac{c_t^a}{v}.
\]
The final quotes are $p_t^b=m_t-\delta_t^{b,\star}$, and $p_t^a=m_t+\delta_t^{a,\star}$. These equations have the same structure as classical Avellaneda--Stoikov quotes. The first term is the liquidity term, determined by the local fill decay parameter. The second term is the inventory continuation-value term. The third term is a local adverse-selection correction. The difference is that the continuation value is not computed for a fixed objective. It is computed through
\[
h_t(q)
=
U_{z_t}(\tau_t,q;\xi_t)^\top z_t,
\]
where \(z_t\) is adapted from recent reward evidence. If the estimated objective implies a positive target inventory, then the value of moving from \(q_t\) to \(q_t+v\) increases. The term $h_t(q_t+v)-h_t(q_t)$ becomes larger, and the bid distance decreases. The market maker buys more aggressively. If recent rewards imply a negative target inventory, the same mechanism shifts aggressiveness to the ask side. If the bid side is toxic, \(c_t^b\) increases and pushes the bid farther away, even if the inventory objective is long-biased.

\section{Online Algorithm}

At each market event, the model first maps the current book state into Avellaneda--Stoikov parameters: $\xi_t=\Xi_\psi(L_t)$.
The recent reward window is then used to estimate the current objective:
$\widehat z_t
=
\mathbb E_{\rho_t}
\left[
r^H(s,a)B_\psi(s,a)
\right]$.
The HJB-forward action-value is
$
Q_{\widehat z_t}(q_t,\tau_t,a;\xi_t)
=
F_{\mathrm{HJB}}(q_t,\tau_t,a,\widehat z_t;\xi_t)^\top \widehat z_t$. The action is selected by
\[
a_t^\star
=
\arg\max_{a\in\mathcal A_{\mathrm{safe}}}
F_{\mathrm{HJB}}(q_t,\tau_t,a,\widehat z_t;\xi_t)^\top \widehat z_t.
\]
In the closed-form Avellaneda--Stoikov version, this maximization is replaced by the quote equations
\[
\delta_t^{b,\star}
=
\frac{1}{\kappa_t^b}
-
\frac{
h_{\widehat z_t}(\tau_t,q_t+v;\xi_t)
-
h_{\widehat z_t}(\tau_t,q_t;\xi_t)
}{v}
+
\frac{c_t^b}{v},
\]
\[
\delta_t^{a,\star}
=
\frac{1}{\kappa_t^a}
-
\frac{
h_{\widehat z_t}(\tau_t,q_t-v;\xi_t)
-
h_{\widehat z_t}(\tau_t,q_t;\xi_t)
}{v}
+
\frac{c_t^a}{v}.
\]
The quotes sent to the market are $p_t^b=m_t-\delta_t^{b,\star}$, and $p_t^a=m_t+\delta_t^{a,\star}$. The essential online loop is therefore
\[
L_t
\longmapsto
\xi_t,
\qquad
\mathcal W_t
\longmapsto
z_t,
\qquad
(\xi_t,z_t)
\longmapsto
h_{z_t},
\qquad
h_{z_t}
\longmapsto
(\delta_t^{b,\star},\delta_t^{a,\star}).
\]
This preserves the speed and interpretability of Avellaneda--Stoikov while allowing the objective to move with the market. Proofs and analyses of the basic properties are given in Appendix \ref{app:basic-proofs}.
%


\section{Experiments}
\label{sec:experiments}
We evaluate three market-making policies on the same reconstructed BTCUSDT
limit-order-book dataset: a Guéant--Lehalle--Fernandez-Tapia (GLFT)
fixed-latency strategy \cite{gueant2017optimal}, a fixed-latency Avellaneda--Stoikov strategy, and our
rolling forward-backward Avellaneda--Stoikov HJB method. Please see full algorithm implementation details in Appendix \ref{sec:implementation}. All experiments use the same market data, preprocessing pipeline, latency model, execution simulator,
fee model, and evaluation window. The purpose of this design is to isolate the
effect of the quoting policy rather than differences in data construction,
latency assumptions, or backtest mechanics.

\subsection{Dataset and Backtest Setup}
\label{subsec:dataset-backtest-setup}
\begin{figure}[t!]
    \centering
    \includegraphics[width=\linewidth]{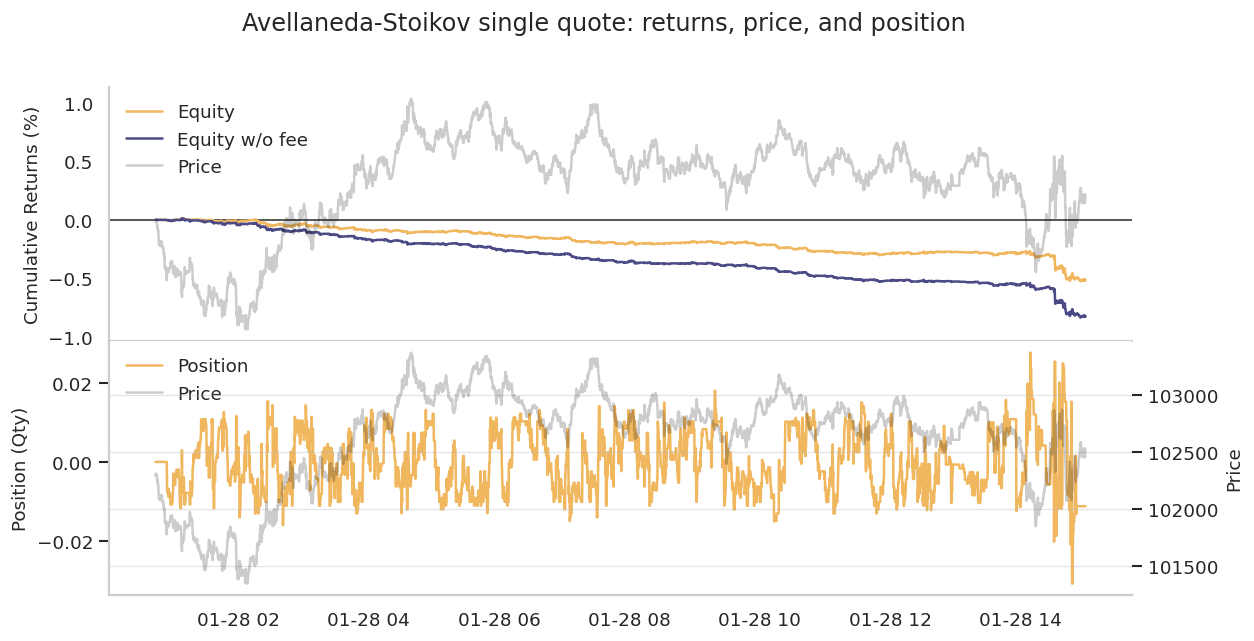}
    \caption{AS single diagnostic.}
    \label{fig:equity-position-summary}
\end{figure}
The evaluation dataset is constructed from raw BTCUSDT order-book and trade
records. We preprocess the feed into the \texttt{hftbacktest}\cite{nkaz001_hftbacktest_2024} event format using
approximate trade reconstruction and depth updates. The final dataset combines
the available January slices and is truncated at $\texttt{2025-01-28 15:00:00 UTC}$
to avoid the later regime break observed in the raw files. The resulting
evaluation feed contains approximately \(12.6\) million events over
\(855.56\) minutes, including approximately \(11.91\) million depth events,
\(678.6\) thousand trade events, and \(8\) thousand snapshots.

We also construct an empirical fixed-latency feed from the available
order-response timestamps. Latencies are capped at one second for numerical
stability. The resulting latency dataset contains \(50{,}090\) observations,
with median entry latency approximately \(570.6\) ms and median response latency
approximately \(427.9\) ms.

All strategies are evaluated in \texttt{hftbacktest}  using the same exchange
model: post-only limit orders, a power-probability queue model, no-partial-fill
exchange mode, tick size \(0.1\), lot size \(10^{-5}\), and the same maker/taker
fee configuration. We use a linear asset accounting model and report standard
high-frequency backtest metrics: cumulative return, Sharpe ratio, Sortino
ratio, maximum drawdown, daily number of trades, daily turnover, return over
maximum drawdown, return per trade, and maximum position value. Maximum drawdown
is reported as a positive magnitude. For each method, we save both summary
equity and position plots and diagnostic plots for volatility, order-arrival
parameters, quote distances, and adaptive variables.

\subsection{Baselines}
\label{subsec:baselines}
The first baseline is the GLFT market-making model. Following the standard
\texttt{hftbacktest} implementation, the strategy estimates order-arrival
intensity parameters from recent market data and quotes around the mid price
using the GLFT reservation-price and spread formulas. We use the same
fixed-latency data as in the other methods. We implemented both the single-quote
and grid variants; the main results table reports the grid variant. The grid
variant places multiple orders around the model-implied quote, while the
single-quote variant submits one bid and one ask. The second baseline is a fixed-latency Avellaneda--Stoikov strategy. This
baseline uses the same rolling volatility and order-arrival estimates as GLFT,
but replaces the GLFT quote equations with the Avellaneda--Stoikov structure.
The model computes an inventory-dependent reservation-price skew and
risk-adjusted spread terms under a Brownian mid-price approximation and Poisson
order arrivals. As with GLFT, we implemented both a single-quote version and a
grid version, and the main results table reports the grid version.
\begin{figure}[t!]
    \centering
    \includegraphics[width=\linewidth]{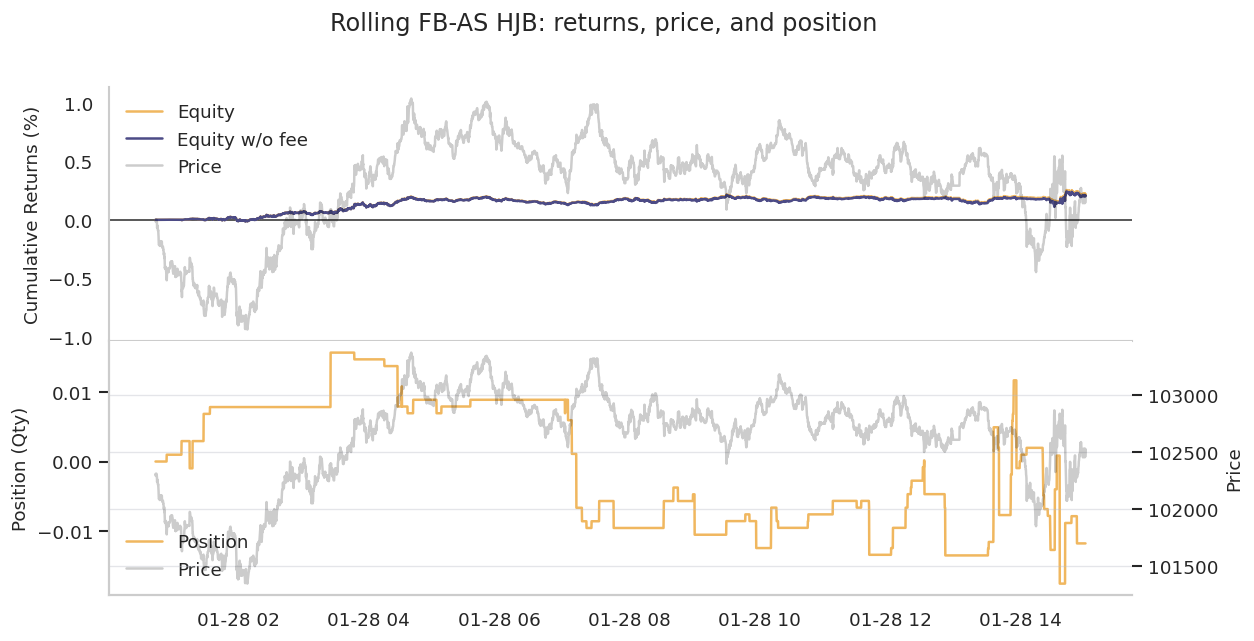}
    \caption{FB-AS diagnostics.}
    \label{fig:fb-as-diagnostics}
\end{figure}
\begin{table}[t!]
    \centering
    \caption{
        Main backtest results on the common BTCUSDT evaluation window.
        Maximum drawdown is reported as a positive magnitude.
    }
    \label{tab:main-results}
    \begin{tabular}{lrrrrr}
        \toprule
        Method
        & Return
        & Sharpe
        & Max Drawdown
        & Daily Trades
        & Daily Turnover \\
        \midrule
        GLFT grid
        & \(-0.001782\)
        & \(-27.88\)
        & \(0.002677\)
        & \(4637.82\)
        & \(46.32\) \\
        AS grid
        & \(-0.001752\)
        & \(-27.26\)
        & \(0.002848\)
        & \(4531.08\)
        & \(45.24\) \\
        Rolling FB-AS HJB
        & \(0.002167\)
        & \(28.57\)
        & \(0.000947\)
        & \(477.85\)
        & \(4.77\) \\
        \bottomrule
    \end{tabular}
\end{table}
\subsection{Results}
\label{subsec:results}

Table~\ref{tab:main-results} reports the main comparison on the common
evaluation window. The grid variants of GLFT and Avellaneda--Stoikov both
produce negative returns, while the rolling FB-AS HJB model produces positive
return and substantially lower drawdown.

The rolling FB-AS HJB strategy trades much less frequently than the GLFT and
Avellaneda--Stoikov grid baselines, but achieves better risk-adjusted
performance on this evaluation window. The reduction in trading intensity
suggests that the adaptive objective and finite-horizon HJB action selection act
as a filter against over-trading under noisy short-horizon conditions. In
particular, the estimated objective dynamically adjusts inventory and
adverse-selection penalties, while the HJB recursion evaluates candidate quotes
through expected future feature occupancy rather than immediate spread capture
alone.

These results should be interpreted as a controlled simulation study rather than
production evidence. The current implementation uses approximate fill
attribution from net position changes, omits explicit fee and impact terms in
the rolling reward labels, and uses reconstructed rather than native full
exchange event data. Nevertheless, because all strategies share the same market
data, latency feed, exchange model, fee model, and preprocessing pipeline, the
comparison is informative about the relative behavior of the quoting rules under
identical backtest conditions.

\section{Conclusion}
\label{sec:conclusion}

We proposed an adaptive market-making framework that preserves the analytical structure of the Avellaneda--Stoikov model while allowing both the market regime and the trading objective to change over time. The central idea is to separate local market dynamics from the objective being optimized. The current limit-order-book state determines the Avellaneda Stoikov parameter vector \(\xi_t\), while recent realized rewards determine an adaptive objective vector \(z_t\). These two quantities are then combined through an HJB-forward representation, and the resulting scalarized value function produces bid and ask quotes. The empirical study on reconstructed BTCUSDT order-book data shows that the rolling FB-AS HJB strategy can behave differently from fixed-objective GLFT and Avellaneda-Stoikov baselines. A production implementation would require native exchange event data, stronger calibration of fill probabilities, more detailed cost modeling, and extensive robustness checks across assets and market regimes.

Several directions remain open. First, the backward representation \(B_\psi\) can be learned more systematically via neural approximations \cite{touati2021learning, asadulaev2026zero} rather than specified through compact hand-crafted coordinates. Second, the market-state map \(\Xi_\psi(L_t)\) can be extended with richer microstructure features. Third, the finite-horizon HJB approximation can be generalized to larger inventory grids, variable order sizes, and multi-asset inventory constraints. Finally, the adaptive objective projection step can be strengthened with explicit risk limits and out-of-distribution penalties.

Overall, the paper shows that zero-shot reward adaptation ideas can be integrated naturally with classical market-making models. By keeping Avellaneda--Stoikov dynamics as the physical core and adapting only a low-dimensional objective vector, the proposed framework provides an interpretable bridge between analytical market making and representation-based reinforcement learning.

\bibliographystyle{plain} 
\bibliography{references}

@article{avellaneda2008high,
  title={High-frequency trading in a limit order book},
  author={Avellaneda, Marco and Stoikov, Sasha},
  journal={Quantitative Finance},
  volume={8},
  number={3},
  pages={217--224},
  year={2008},
  publisher={Taylor \& Francis}
}

@article{gueant2017optimal,
  title={Optimal market making},
  author={Gu{\'e}ant, Olivier},
  journal={Applied Mathematical Finance},
  volume={24},
  number={2},
  pages={112--154},
  year={2017},
  publisher={Taylor \& Francis}
}

@article{touati2021learning,
  title={Learning one representation to optimize all rewards},
  author={Touati, Ahmed and Ollivier, Yann},
  journal={Advances in Neural Information Processing Systems},
  volume={34},
  pages={13--23},
  year={2021}
}

@software{nkaz001_hftbacktest_2024,
  author       = {nkaz001 https://github.com/nkaz001/hftbacktest },
  title        = {HftBacktest},
  year         = {2024},
  url          = {https://github.com/nkaz001/hftbacktest},
  version      = {2.x}
}

@article{asadulaev2026zero,
  title={Zero-Shot Off-Policy Learning},
  author={Asadulaev, Arip and Bobrin, Maksim and Lahlou, Salem and Dylov, Dmitry and Karray, Fakhri and Takac, Martin},
  journal={arXiv preprint arXiv:2602.01962},
  year={2026}
}

@article{tirinzoni2025zero,
  title={Zero-shot whole-body humanoid control via behavioral foundation models},
  author={Tirinzoni, Andrea and Touati, Ahmed and Farebrother, Jesse and Guzek, Mateusz and Kanervisto, Anssi and Xu, Yingchen and Lazaric, Alessandro and Pirotta, Matteo},
  journal={arXiv preprint arXiv:2504.11054},
  year={2025}
}
\newpage
\appendix

\section{Practical Specifications}
\subsection{Stabilization and Corrections}

The pure FB estimate
\[
\widehat z_t
=
\mathbb E_{\rho_t}
\left[
r^H(s,a)B_\psi(s,a)
\right]
\]
can be noisy, especially in short windows. A stable production version should therefore mix it with a safe prior objective:
\[
z_{\mathrm{prior}}
=
\begin{pmatrix}
1\\
0\\
-\lambda_0\\
-\nu_0
\end{pmatrix}.
\]
This prior corresponds to an inventory-neutral Avellaneda--Stoikov objective. The adapted objective is
\[
z_t
=
(1-\beta)z_{\mathrm{prior}}
+
\beta \widehat z_t.
\]
A projection can then enforce economically meaningful constraints:
\[
z_t
=
\Pi_{\mathcal Z}
\left[
(1-\beta)z_{\mathrm{prior}}
+
\beta \widehat z_t
\right],
\]
where
\[
\mathcal Z
=
\left\{
z:
z_{\mathrm{pnl}}=1,\ 
z_{q^2}\le -\lambda_{\min},\
z_{\mathrm{adv}}\le 0,\
|\theta(z)|\le Q_{\max}
\right\}.
\]
The implied target inventory is
\[
\theta(z)
=
\frac{z_q}{-2z_{q^2}}.
\]
The projection prevents the adaptive objective from becoming risk-seeking after a short lucky period.

If the backward representation is not perfectly whitened, the estimate
\[
\widehat z_t=\mathbb E_{\rho_t}[r^H B]
\]
can be replaced by a ridge projection. Let
\[
b_i=B_\psi(s_i,a_i),
\]
\[
C_t=
\sum_{i\in\mathcal W_t}
w_{t,i}b_i b_i^\top,
\qquad
u_t=
\sum_{i\in\mathcal W_t}
w_{t,i}r_i^H b_i.
\]
Then
\[
\widehat z_t
=
(C_t+\lambda_z I)^{-1}u_t.
\]
This is equivalent to the regularized local regression
\[
\widehat z_t
=
\arg\min_z
\sum_{i\in\mathcal W_t}
w_{t,i}
(r_i^H-z^\top b_i)^2
+
\lambda_z\|z\|^2.
\]

The weights may also be state-conditioned rather than purely time-decayed:
\[
w_{t,i}
\propto
\exp\left(-\frac{t-i}{\ell}\right)
K(L_t,L_i).
\]
For example,
\[
K(L_t,L_i)
=
\exp\left(
-\frac{
\|E_\psi(L_t)-E_\psi(L_i)\|^2
}{2h^2}
\right).
\]
This makes the objective estimate local not only in time but also in market-regime space:
\[
z_t
\approx
\mathbb E_{\rho(\cdot|L_t)}
\left[
r^H(s,a)B_\psi(s,a)
\right].
\]

A further stabilizer is exponential smoothing:
\[
z_t^{\mathrm{smooth}}
=
(1-\alpha_z)z_{t-1}^{\mathrm{smooth}}
+
\alpha_z z_t.
\]
The HJB-forward is then evaluated with \(z_t^{\mathrm{smooth}}\) instead of the raw estimate.

When the dominant adaptation is target inventory, it is useful to work with centered inventory:
\[
\widetilde q_t=q_t-\theta_t.
\]
The inventory penalty becomes
\[
-\lambda_t(q_t-\theta_t)^2
=
-\lambda_t \widetilde q_t^2.
\]
The value function can then be represented as
\[
h_{z_t}(\tau,q;\xi)
\approx
\widetilde h_{\lambda_t,\nu_t}(\tau,q-\theta_t;\xi).
\]
The quote equations become
\[
\delta_t^{b,\star}
=
\frac{1}{\kappa_t^b}
-
\frac{
\widetilde h_t(\widetilde q_t+v)
-
\widetilde h_t(\widetilde q_t)
}{v}
+
\frac{c_t^b}{v},
\]
\[
\delta_t^{a,\star}
=
\frac{1}{\kappa_t^a}
-
\frac{
\widetilde h_t(\widetilde q_t-v)
-
\widetilde h_t(\widetilde q_t)
}{v}
+
\frac{c_t^a}{v}.
\]
This is the most direct Avellaneda--Stoikov interpretation of target-inventory adaptation: the model replaces \(q_t\) by \(q_t-\theta_t\).

Finally, the chosen action can be constrained by an out-of-distribution penalty:
\[
a_t^\star
=
\arg\max_{a\in\mathcal A_{\mathrm{safe}}}
\left[
F_{\mathrm{HJB}}(q_t,\tau_t,a,z_t;\xi_t)^\top z_t
-
\lambda_{\mathrm{OOD}}D_{\mathrm{OOD}}(a|L_t)
\right].
\]
After quotes are computed, hard safety constraints are applied:
\[
\delta_{\min}
\le
\delta_t^{b/a}
\le
\delta_{\max},
\qquad
|q_t|\le Q_{\max},
\qquad
v_t^{b/a}\le v_{\max}.
\]
If inventory is close to a hard limit, one side of quoting can be disabled:
\[
q_t\ge Q_{\max}
\Rightarrow
v_t^b=0,
\qquad
q_t\le -Q_{\max}
\Rightarrow
v_t^a=0.
\]

\subsection{Final Specification}

The adaptive FB-AS market maker is fully summarized by the following equations:
\[
\xi_t=\Xi_\psi(L_t),
\]
\[
z_t
=
\Pi_{\mathcal Z}
\left[
(1-\beta)z_{\mathrm{prior}}
+
\beta
\mathbb E_{\rho_t}
\left[
r^H(s,a)B_\psi(s,a)
\right]
\right],
\]
\[
a_t^\star
=
\arg\max_{a\in\mathcal A_{\mathrm{safe}}}
F_{\mathrm{HJB}}(q_t,\tau_t,a,z_t;\xi_t)^\top z_t,
\]
\[
p_t^b=m_t-\delta_t^{b,\star},
\qquad
p_t^a=m_t+\delta_t^{a,\star}.
\]
In the closed-form Avellaneda--Stoikov implementation,
\[
\delta_t^{b,\star}
=
\frac{1}{\kappa_t^b}
-
\frac{
h_{z_t}(\tau_t,q_t+v;\xi_t)
-
h_{z_t}(\tau_t,q_t;\xi_t)
}{v}
+
\frac{c_t^b}{v},
\]
\[
\delta_t^{a,\star}
=
\frac{1}{\kappa_t^a}
-
\frac{
h_{z_t}(\tau_t,q_t-v;\xi_t)
-
h_{z_t}(\tau_t,q_t;\xi_t)
}{v}
+
\frac{c_t^a}{v},
\]
where
\[
h_{z_t}(\tau,q;\xi_t)
=
U_{z_t}(\tau,q;\xi_t)^\top z_t.
\]

The model is therefore an Avellaneda--Stoikov market maker with adaptive objectives. The market regime \(\xi_t\) determines the local dynamics, the FB-style vector \(z_t\) determines what the market maker currently wants, and the HJB-forward representation converts the pair \((\xi_t,z_t)\) into optimal quotes. The analytical structure of Avellaneda--Stoikov is preserved, but the objective is no longer static. It is continuously inferred from recent reward evidence and projected back into a safe, interpretable family of trading objectives.

\section{Proofs of Basic Properties}
\label{app:basic-proofs}

This appendix collects elementary properties used in the adaptive Avellaneda--Stoikov construction. 
Throughout, we use the feature vector
\[
\phi(s,a,s')
=
\begin{pmatrix}
\Delta W\\
q'\\
(q')^2\\
\mathrm{adv}
\end{pmatrix},
\]
and the linear reward
\[
r_z(s,a,s')=z^\top\phi(s,a,s').
\]

\begin{lemma}[Target-inventory rewards are linear in the feature vector]
\label{lem:target-linear-app}
Let
\[
\phi(s,a,s')=
\begin{pmatrix}
\Delta W\\
q'\\
(q')^2\\
\mathrm{adv}
\end{pmatrix}.
\]
For any inventory target \(\theta\), inventory penalty \(\lambda\ge 0\), and adverse-selection penalty \(\nu\ge 0\), the reward
\[
r_\theta(s,a,s')
=
\Delta W-\lambda(q'-\theta)^2-\nu\,\mathrm{adv}
\]
is equivalent for action selection to the linear reward \(r_z(s,a,s')=z^\top\phi(s,a,s')\) with
\[
z=
\begin{pmatrix}
1\\
2\lambda\theta\\
-\lambda\\
-\nu
\end{pmatrix}.
\]
\end{lemma}

\begin{proof}
Expanding the quadratic inventory penalty gives
\[
-\lambda(q'-\theta)^2
=
-\lambda(q')^2+2\lambda\theta q'-\lambda\theta^2.
\]
Therefore,
\[
r_\theta(s,a,s')
=
\Delta W
+
2\lambda\theta q'
-
\lambda(q')^2
-
\nu\,\mathrm{adv}
-
\lambda\theta^2.
\]
The final term \(-\lambda\theta^2\) does not depend on the action. Hence it does not change the action that maximizes the reward. Thus, for action selection,
\[
r_\theta(s,a,s')
\sim
\Delta W
+
2\lambda\theta q'
-
\lambda(q')^2
-
\nu\,\mathrm{adv}.
\]
This is exactly \(z^\top\phi(s,a,s')\) with
\[
z=
\begin{pmatrix}
1\\
2\lambda\theta\\
-\lambda\\
-\nu
\end{pmatrix}.
\]
\end{proof}

\begin{lemma}[Implied risk penalty and target inventory]
\label{lem:implied-target-app}
Suppose \(z_{q^2}<0\). If the adaptive objective vector \(z\) is interpreted as a target-inventory objective, then its implied inventory penalty and target inventory are
\[
\lambda(z)=-z_{q^2},
\qquad
\theta(z)=\frac{z_q}{-2z_{q^2}}.
\]
\end{lemma}

\begin{proof}
From Lemma~\ref{lem:target-linear-app}, a target-inventory objective has inventory coordinates
\[
z_q=2\lambda\theta,
\qquad
z_{q^2}=-\lambda.
\]
Therefore,
\[
\lambda=-z_{q^2}.
\]
Since \(z_{q^2}<0\), this gives \(\lambda>0\). Substituting \(\lambda=-z_{q^2}\) into \(z_q=2\lambda\theta\) gives
\[
z_q
=
2(-z_{q^2})\theta.
\]
Hence
\[
\theta
=
\frac{z_q}{-2z_{q^2}}.
\]
\end{proof}

\begin{proposition}[Closed form of the ridge objective estimator]
\label{prop:ridge-estimator-app}
Let \(b_i\in\mathbb R^d\), \(r_i^H\in\mathbb R\), \(w_i\ge 0\), and \(\lambda_z>0\). The regularized objective-estimation problem
\[
\widehat z_t
=
\arg\min_z
\sum_{i\in W_t} w_i(r_i^H-z^\top b_i)^2+\lambda_z\|z\|_2^2
\]
has the unique solution
\[
\widehat z_t
=
(C_t+\lambda_z I)^{-1}u_t,
\]
where
\[
C_t=\sum_{i\in W_t}w_i b_i b_i^\top,
\qquad
u_t=\sum_{i\in W_t}w_i r_i^H b_i.
\]
\end{proposition}

\begin{proof}
Define
\[
J(z)
=
\sum_{i\in W_t} w_i(r_i^H-z^\top b_i)^2+\lambda_z\|z\|_2^2.
\]
Differentiating with respect to \(z\) gives
\[
\nabla J(z)
=
-2\sum_{i\in W_t}w_i(r_i^H-z^\top b_i)b_i
+
2\lambda_z z.
\]
Rearranging,
\[
\nabla J(z)
=
2\left(\sum_{i\in W_t}w_i b_i b_i^\top+\lambda_z I\right)z
-
2\sum_{i\in W_t}w_i r_i^H b_i.
\]
Using the definitions
\[
C_t=\sum_{i\in W_t}w_i b_i b_i^\top,
\qquad
u_t=\sum_{i\in W_t}w_i r_i^H b_i,
\]
the first-order condition \(\nabla J(z)=0\) becomes
\[
(C_t+\lambda_z I)z=u_t.
\]
Since \(\lambda_z>0\), the matrix \(C_t+\lambda_z I\) is positive definite. Indeed, for any nonzero \(x\),
\[
x^\top(C_t+\lambda_z I)x
=
\sum_{i\in W_t}w_i(x^\top b_i)^2+\lambda_z\|x\|_2^2
>
0.
\]
Thus \(C_t+\lambda_z I\) is invertible, and the unique minimizer is
\[
\widehat z_t
=
(C_t+\lambda_z I)^{-1}u_t.
\]
\end{proof}

\begin{theorem}[Scalar optimality of the finite-horizon vector HJB]
\label{thm:vector-hjb-scalar-optimality-app}
Fix \(z\) and \(\xi\). Assume that the inventory grid and safe action set are finite. Let \(U_0(q)=0\), and for \(n=1,\ldots,N\), define
\[
F_n(q,a,z;\xi)
=
\bar\phi(q,a;\xi)
+
\beta\,\mathbb E_\xi[U_{n-1}(q')\mid q,a],
\]
\[
a_n^\star(q)
\in
\arg\max_{a\in\mathcal A_{\mathrm{safe}}}
F_n(q,a,z;\xi)^\top z,
\]
and
\[
U_n(q)=F_n(q,a_n^\star(q),z;\xi).
\]
Then
\[
V_n(q):=U_n(q)^\top z
\]
is the optimal \(n\)-step value function for the scalar reward \(r_z=z^\top\phi\).
\end{theorem}

\begin{proof}
We prove the result by induction on \(n\).

For \(n=0\), there are no future rewards, so
\[
U_0(q)=0.
\]
Therefore,
\[
V_0(q)=U_0(q)^\top z=0,
\]
which is the correct zero-horizon scalar value.

Now suppose the claim holds for horizon \(n-1\). That is,
\[
V_{n-1}(q)=U_{n-1}(q)^\top z
\]
is the optimal \((n-1)\)-step scalar value. For any candidate action \(a\),
\[
F_n(q,a,z;\xi)^\top z
=
\left[
\bar\phi(q,a;\xi)
+
\beta\,\mathbb E_\xi[U_{n-1}(q')\mid q,a]
\right]^\top z.
\]
Using linearity of the inner product and expectation,
\[
F_n(q,a,z;\xi)^\top z
=
\bar\phi(q,a;\xi)^\top z
+
\beta\,\mathbb E_\xi[U_{n-1}(q')^\top z\mid q,a].
\]
Since \(r_z=z^\top\phi\), the expected immediate scalar reward is
\[
\bar r_z(q,a;\xi)
=
\bar\phi(q,a;\xi)^\top z.
\]
By the induction hypothesis,
\[
U_{n-1}(q')^\top z=V_{n-1}(q').
\]
Thus
\[
F_n(q,a,z;\xi)^\top z
=
\bar r_z(q,a;\xi)
+
\beta\,\mathbb E_\xi[V_{n-1}(q')\mid q,a].
\]
The action \(a_n^\star(q)\) is chosen to maximize this quantity over \(a\in\mathcal A_{\mathrm{safe}}\). Therefore,
\[
V_n(q)
=
U_n(q)^\top z
=
\max_{a\in\mathcal A_{\mathrm{safe}}}
\left\{
\bar r_z(q,a;\xi)
+
\beta\,\mathbb E_\xi[V_{n-1}(q')\mid q,a]
\right\}.
\]
This is exactly the Bellman optimality recursion for the \(n\)-step scalar reward problem. Since the safe action set is finite, the maximum exists. Hence \(V_n(q)=U_n(q)^\top z\) is the optimal \(n\)-step scalar value.
\end{proof}

\begin{theorem}[Closed-form Avellaneda--Stoikov quote under exponential fills]
\label{thm:as-closed-form-app}
Let \(A>0\), \(\kappa>0\), \(v>0\), continuation-value difference \(D\), and adverse-selection cost \(c\). For the one-sided objective
\[
J(\delta)
=
Ae^{-\kappa\delta}(v\delta+D-c),
\]
the unconstrained maximizer is
\[
\delta^\star
=
\frac{1}{\kappa}
-
\frac{D}{v}
+
\frac{c}{v}.
\]
Consequently, with
\[
D^b=h(q+v)-h(q),
\qquad
D^a=h(q-v)-h(q),
\]
the bid and ask quote distances are
\[
\delta^{b,\star}
=
\frac{1}{\kappa^b}
-
\frac{h(q+v)-h(q)}{v}
+
\frac{c^b}{v},
\]
and
\[
\delta^{a,\star}
=
\frac{1}{\kappa^a}
-
\frac{h(q-v)-h(q)}{v}
+
\frac{c^a}{v}.
\]
\end{theorem}

\begin{proof}
Differentiate
\[
J(\delta)=Ae^{-\kappa\delta}(v\delta+D-c)
\]
with respect to \(\delta\). We obtain
\[
J'(\delta)
=
Ae^{-\kappa\delta}
\left[
v-\kappa(v\delta+D-c)
\right].
\]
Since \(Ae^{-\kappa\delta}>0\), the first-order condition is
\[
v-\kappa(v\delta+D-c)=0.
\]
Solving for \(\delta\),
\[
\kappa v\delta
=
v-\kappa D+\kappa c,
\]
and therefore
\[
\delta^\star
=
\frac{1}{\kappa}
-
\frac{D}{v}
+
\frac{c}{v}.
\]

The term
\[
v-\kappa(v\delta+D-c)
\]
is strictly decreasing in \(\delta\). Hence \(J'(\delta)>0\) for \(\delta<\delta^\star\) and \(J'(\delta)<0\) for \(\delta>\delta^\star\). Therefore \(J\) increases up to \(\delta^\star\) and decreases after \(\delta^\star\), so \(\delta^\star\) is the unconstrained maximizer.

For the bid side, a fill moves inventory from \(q\) to \(q+v\). Hence the continuation-value difference is
\[
D^b=h(q+v)-h(q).
\]
Substituting \(D=D^b\), \(\kappa=\kappa^b\), and \(c=c^b\) gives
\[
\delta^{b,\star}
=
\frac{1}{\kappa^b}
-
\frac{h(q+v)-h(q)}{v}
+
\frac{c^b}{v}.
\]
For the ask side, a fill moves inventory from \(q\) to \(q-v\). Hence
\[
D^a=h(q-v)-h(q).
\]
Substituting \(D=D^a\), \(\kappa=\kappa^a\), and \(c=c^a\) gives
\[
\delta^{a,\star}
=
\frac{1}{\kappa^a}
-
\frac{h(q-v)-h(q)}{v}
+
\frac{c^a}{v}.
\]
This proves the closed-form quote equations.
\end{proof}

\paragraph{Remark on quote constraints.}
If quote distances are constrained to an interval
\[
\delta\in[\delta_{\min},\delta_{\max}],
\]
then the constrained quote is simply the clipped value
\[
\delta_{\mathrm{safe}}^\star
=
\Pi_{[\delta_{\min},\delta_{\max}]}
\left(
\frac{1}{\kappa}
-
\frac{D}{v}
+
\frac{c}{v}
\right).
\]
This corresponds to the hard safety constraints used by the online implementation.

\section{Implementation Details}
\label{sec:implementation}
\subsection{Rolling FB-AS HJB Method}
\label{subsec:rolling-fb-as-hjb-method}
Our method is a rolling forward-backward Avellaneda--Stoikov market maker with a
finite-horizon inventory-grid HJB. It is designed to match the paper's online
architecture:
\[
L_t \longmapsto \xi_t,
\qquad
\mathcal W_t \longmapsto z_t,
\qquad
a_t^\star
=
\arg\max_{a\in\mathcal A_{\mathrm{safe}}}
F_{\mathrm{HJB}}(q_t,\tau_t,a,z_t;\xi_t)^\top z_t .
\]
Here \(L_t\) denotes the recent limit-order-book state, \(\xi_t\) denotes the
local market parameters, \(\mathcal W_t\) is a rolling window of completed
markout observations, \(z_t\) is the adaptive objective vector, and
\(F_{\mathrm{HJB}}\) is the forward HJB feature map used for action selection.

At each decision time, the strategy estimates local volatility, order-arrival
intensities, depth-decay parameters, and adverse-selection proxies from recent
data. These quantities define the local Avellaneda--Stoikov parameter vector
\[
\xi_t
=
\left(
\sigma_t,
A_t^b,
A_t^a,
\kappa_t^b,
\kappa_t^a,
c_t^b,
c_t^a
\right).
\]
The action is a pair of bid and ask quote distances, $a_t=(\delta_t^b,\delta_t^a)$, 
chosen from a discrete safe action grid. The safe action set
\(\mathcal A_{\mathrm{safe}}\) enforces the Avellaneda--Stoikov quote floor,
post-only constraints, and prospective inventory limits.

For a candidate action \(a=(\delta^b,\delta^a)\), the fill intensities are
\[
\lambda^b(\delta^b)
=
A_t^b e^{-\kappa_t^b\delta^b},
\qquad
\lambda^a(\delta^a)
=
A_t^a e^{-\kappa_t^a\delta^a}.
\]
Over one HJB time step \(\Delta t\), the corresponding fill probabilities are
\[
p_b
=
1-\exp\!\left(-\lambda^b(\delta^b)\Delta t\right),
\qquad
p_a
=
1-\exp\!\left(-\lambda^a(\delta^a)\Delta t\right).
\]
If a side is disabled by the inventory constraint, its corresponding fill
probability is set to zero.

The immediate expected feature vector is
\[
\bar\phi(q,a;\xi_t)
=
\mathbb E_{\xi_t}
\left[
\begin{pmatrix}
\Delta W_{t+1}\\
q_{t+1}\\
q_{t+1}^2\\
\mathrm{adv}_{t+1}
\end{pmatrix}
\middle| q_t=q,\ a_t=a
\right].
\]
In the implementation, the first coordinate is the expected spread/wealth
capture, while the adverse-selection term is kept as a separate proxy feature.
This avoids subtracting adverse cost inside the PnL coordinate and then
penalizing it again through \(z_{\mathrm{adv}}\).

Inventory evolves according to
\[
q' = q + N_b - N_a,
\]
where \(N_b\) and \(N_a\) are independent Bernoulli fill indicators with
probabilities \(p_b\) and \(p_a\). Therefore the inventory-risk coordinate uses
\[
\mathbb E[(q')^2]
=
q^2
+
2q(p_b-p_a)
+
p_b+p_a-2p_bp_a,
\]
rather than \((\mathbb E[q'])^2\).

\subsubsection{Forward HJB representation}

For the forward component, we compute a finite-horizon inventory-grid HJB
recursion. The value representation is vector-valued: each inventory state
stores expected future feature occupancy rather than a scalar value. We set the
terminal feature vector to zero,
\[
U_0(q;z_t,\xi_t)=0,
\]
and for \(n=1,\ldots,N\) compute
\[
F_n(q,a,z_t;\xi_t)
=
\bar\phi(q,a;\xi_t)
+
\beta
\mathbb E_{\xi_t}
\left[
U_{n-1}(q';z_t,\xi_t)
\mid q,a
\right],
\]
\[
a_n^\star(q)
=
\arg\max_{a\in\mathcal A_{\mathrm{safe}}}
F_n(q,a,z_t;\xi_t)^\top z_t,
\]
\[
U_n(q;z_t,\xi_t)
=
F_n(q,a_n^\star(q),z_t;\xi_t).
\]
The live policy uses the final HJB layer:
\[
a_t^\star
=
\arg\max_{a\in\mathcal A_{\mathrm{safe}}}
F_N(q_t,a,z_t;\xi_t)^\top z_t .
\]

The transition expectation includes four independent fill outcomes:
\[
p_{00}=(1-p_b)(1-p_a),
\qquad
p_{10}=p_b(1-p_a),
\]
\[
p_{01}=(1-p_b)p_a,
\qquad
p_{11}=p_bp_a .
\]
The both-fill event returns inventory to the same state \(q\). Therefore the
continuation expectation is evaluated as
\[
\mathbb E[U_{n-1}(q')]
=
(p_{00}+p_{11})U_{n-1}(q)
+
p_{10}U_{n-1}(q+1)
+
p_{01}U_{n-1}(q-1).
\]
In each online HJB solve, the current market parameter vector \(\xi_t\) is held
fixed over the finite horizon. This is a receding-horizon approximation to the
nonstationary market-making problem.

\subsubsection{Backward objective estimation}

The backward component estimates the adaptive objective vector \(z_t\) from
realized markout rewards. We use the objective coordinates
\[
z_t
=
\left(
1,
z_{q,t},
z_{q^2,t},
z_{\mathrm{adv},t}
\right),
\]
fixing \(z_{\mathrm{pnl},t}=1\) as in the safe objective family. The remaining
coordinates are estimated from a rolling window of completed fill observations.

For each filled quote \(i\), we store a fill-conditional backward row $
b_i
=
\left(
b_{0,i},
b_{1,i},
b_{2,i},
b_{3,i}
\right),
$
where \(b_{0,i}\) is spread/wealth capture, \(b_{1,i}\) is next inventory,
\(b_{2,i}\) is squared next inventory with the same risk scaling used in the
forward map, and \(b_{3,i}\) is the adverse-cost proxy. After a fixed markout
horizon \(H\), the realized label is
\[
r_i^H
=
N_i^b v_i^b
\left(
\delta_i^b + m_{i+H}-m_i
\right)
+
N_i^a v_i^a
\left(
\delta_i^a - (m_{i+H}-m_i)
\right).
\]
In this compact prototype, fee and impact terms are omitted from the adaptive
reward labels, although the backtest performance accounting uses the common
simulator fee model.

Because the wealth coordinate is fixed to one, we estimate only the residual
objective coordinates by ridge regression:
\[
r_i^H-b_{0,i}
=
z_{q,t} b_{1,i}
+
z_{q^2,t} b_{2,i}
+
z_{\mathrm{adv},t} b_{3,i}
+
\epsilon_i .
\]
Let
\[
x_i=
\begin{pmatrix}
b_{1,i}\\
b_{2,i}\\
b_{3,i}
\end{pmatrix},
\qquad
y_i=r_i^H-b_{0,i}.
\]
The rolling ridge estimate is
$
\widehat z_{\mathrm{rest},t}
=
(C_t+\lambda_z I)^{-1}u_t,
$
where
\[
C_t=\sum_{i\in\mathcal W_t} w_{t,i}x_i x_i^\top,
\qquad
u_t=\sum_{i\in\mathcal W_t} w_{t,i}y_i x_i .
\]
The estimate is mixed with a safe prior and projected into an economically
meaningful set:
\[
z_{q^2,t}\le 0,
\qquad
z_{\mathrm{adv},t}\le 0,
\qquad
|\theta_t|\le Q_{\max}.
\]
The implied target inventory is $
\theta_t
=
\frac{z_{q,t}}
{-2z_{q^2,t}},
$
with the implementation using the corresponding normalized and risk-scaled
version of this expression. This projection prevents short-run noisy reward
observations from producing risk-seeking objectives.

\end{document}